\documentclass{article}
\usepackage[utf8]{inputenc}

\usepackage{amsmath}
\usepackage{amssymb}
\usepackage{amsfonts}
\usepackage{amsthm}
\usepackage{mathtools}
\usepackage{bbm}
\usepackage{bm}
\usepackage{xcolor}
\usepackage{ulem}

\newtheorem{proposition}{Proposition}

\newcommand{\cube}[1]{\mathbb{F}_2^{#1}}

\newcommand{\R}{\mathbb{R}}

\newcommand{\1}{\mathbbm{1}}

\DeclareMathOperator{\supp}{supp}

\renewcommand\em{\it}

\title{An Elementary Proof
of the First LP Bound\\on the Rate of Binary Codes}
\author{Nati Linial\thanks{School of Computer Science and Engineering, Hebrew University, Jerusalem 91904, Israel. e-mail: nati@cs.huji.ac.il. Supported in part by a NSF-BSF research grant "Global Geometry of Graphs".}
\and{Elyassaf Loyfer\thanks{School of Computer Science and Engineering, Hebrew University, Jerusalem 91904, Israel. e-mail: elyassaf.loyfer@mail.huji.ac.il.}}}

\begin{document}

\maketitle

\begin{abstract}
    The asymptotic rate vs.\ distance problem is a long-standing fundamental problem in coding theory.
    The best upper bound to date 
    was given in 1977, and has received since then
    numerous proofs and interpretations.
    Here we provide a new, elementary proof of
    this bound that is
    based on counting walks in the Hamming cube.
\end{abstract}

\section{Introduction}

A binary code of length $n$ is a subset $C \subset \cube{n}$. 
Its distance is the smallest Hamming distance
among all pairs of code words, 
$dist(C) \coloneqq \min_{x\neq y \in C} |x+y|$.
A fundamental problem in coding theory is to find the largest possible size, $A(n,d)$, of
a code of length $n$ and distance $d$:
\[
    A(n,d) \coloneqq \max \{|C| : C\subset\cube{n},~dist(C) \geq d\} 
\]
The asymptotic version of this question, the {\em rate vs.\ distance problem},
seeks the largest possible rate $R(C)\coloneqq n^{-1}\log_n(|C|)$, where
the distance is linear in $n$:
\[
    \mathcal{R}(\delta) \coloneqq
    \limsup_{n\to\infty} 
    \{ n^{-1} \log_2 \left( A(n,\lfloor \delta n \rfloor) \right) \}
\]
It is known that $\mathcal{R}(0) = 1$ and $\mathcal{R}(\delta) = 0$
when $\delta \geq 1/2$, and $\mathcal{R}(\delta) > 0$ when $\delta < 1/2$. 
The value of $\mathcal{R}(\delta)$ remains unknown for all $\delta\in(0,1/2)$.

The best and only lower bound, 
$\mathcal{R}(\delta) \geq 1-H(\delta)$, 
is due to Gilbert \cite{gilbert1952comparison} 
and Varshamov \cite{varshamov1957estimate}.
Here, $H(\cdot)$ is the binary entropy function.
The best upper bounds that we have on $\mathcal{R}(\delta)$
are the first and second linear programming (LP) bounds. These
are due to McEliece, Rodemich,
Rumsey and Welch (MRRW, \cite{mceliece1977new}).
The first 
is the better of the two when $\delta \in (0.273,1/2)$,
\begin{equation}
    \mathcal{R}(\delta)\leq H(1/2 - \sqrt{\delta(1-\delta)})
    \label{eq:first_lp_bound}
\end{equation}
The second one,
whose statement is more involved dominates when $\delta\in(0,0.273)$.

These bounds are 
based on a linear program due to Delsarte \cite{delsarte1973algebraic}, 
the optimum of which is an upper bound on $A(n,d)$. 
Although this maximum is still unknown, 
upper bounds on $A(n,d)$ can be derived by passing to the LP dual and 
constructing a dual feasible solution.

In the this paper we present a new, elementary proof 
of \eqref{eq:first_lp_bound}, the first LP bound.

\section{Preliminaries}

As usual, the Hamming cube is the graph whose vertex set is $\cube{n}$, 
where adjacency means having Hamming distance one. The $j$-th
{\em level} of the cube, for $0\leq j\leq n$, is the set of vertices of Hamming weight $j$.

Let $f,g:\cube{n}\to\R$ be real functions on the cube.
We define their inner product w.r.t.\ the uniform measure,
$\langle f,g\rangle \coloneqq 2^{-n} \sum_{x\in\cube{n}}f(x)g(x)$.

To every $x\in\cube{n}$ there corresponds a {\em Fourier character} $\chi_x:\cube{n}\to\R$ that
is defined via $\chi_x(y) = (-1)^{\langle x,y\rangle}$. The 
set of characters $\{\chi_x\}_{x\in\cube{n}}$ forms an
orthonormal basis of the space of real functions on the 
Hamming cube.
The Fourier transform of $f$ is defined by
\[
    \widehat{f}(x) \coloneqq \langle \chi_x, f\rangle 
    = 2^{-n}\sum_{y\in\cube{n}} (-1)^{\langle x,y\rangle} f(y)
\]
In Fourier space, inner product is defined without normalization:
$\langle \widehat{f},\widehat{g}\rangle \coloneqq \sum_{x\in\cube{n}}
\widehat{f}(x)\widehat{g}(x)$.
Parseval's theorem states that 
$\langle f,g\rangle = \langle \widehat{f},\widehat{g}\rangle$.
Convolution between $f$ and $g$ is defined by
\[
    (f*g)(x) \coloneqq 2^{-n}\sum_{y\in\cube{n}} f(y)g(x+y)
\]
By the convolution theorem,
$
    \widehat{f* g} = \widehat{f} \cdot \widehat{g}
$.

\section{Upper Bound on $A(n,d)$ through Delsarte's LP}
\label{section:delsarte_lp}

In this section we recall the dual of Delsarte’s LP for binary codes and
describe a general method to construct feasible dual solutions.
Delsarte's original work was formulated in the context of association schemes
and generalizes to other metric spaces. 
For binary codes with the Hamming distance, the LP can be developed 
completely through Fourier analysis.

Consider a binary code $C\subset\cube{n}$
of length $n$ and minimal distance $d$. 
Let $\1_C:\cube{n}\to\R$ be its 
indicator function, and let
$f_C \coloneqq \frac{2^n}{|C|}\1_C * \1_C$.
The following properties of $f_C$ are
easy to verify:
\[
    f_C(0)=1;~ 
    f_C\geq 0;~ 
    f_C(x) = 0 ~ \text{if}~ 1\leq |x|\leq d-1; ~
    \text{and}~ \widehat{f}_C\geq 0.
\] The last property
follows from the convolution theorem. 
The sum of $f_C$ over the entire cube gives the 
cardinality of $C$. This yields {\em Delsarte's LP}.
\begin{equation}   
\begin{split}
\max 2^n\langle f,\1\rangle~\text{over all}~f:\cube{n} \to\mathbb{R}\text{~with~}\\  
f(0)=1,~ f\ge 0,~ \widehat{f}\ge 0,~ f(x) = 0\text{~if~} 1\leq |x|\leq d-1
\end{split}
\end{equation}
Consequently, the maximum of this LP is an upper bound on $A(n,d)$.

\hfill

Passing to the LP dual, we say that a function $g:\cube{n}\to\R$
is {\em feasible} if
\begin{equation}
\widehat{g} \geq 0,~\widehat{g}(0)>0,~ g(x) \leq 0~\text{if}~|x|\geq d
\label{eq:dual_feasible}
\end{equation}
If $g$ is feasible, then $g(0)/\widehat{g}(0)$ gives an upper bound on $|C|$,
\[
    \widehat{g}(0)|C|
    = 2^n \widehat{g}(0) \widehat{f}_C(0)
    \leq 2^n \langle \widehat{g},\widehat{f}_C\rangle
    = 2^n \langle f_C,g \rangle
    \leq f_C(0) g(0) = g(0)
\]
Explanation: for the first equality, $2^n\widehat{f}_C(0)= \sum_{x\in\cube{n}}f_C(x)=|C|$.
The first inequality holds because $\widehat{f}_C,\widehat{g}\geq 0$.
The subsequent equality follows from Parseval's identity.
For the second inequality, note that $f_C$ vanishes when $1\leq |x|\leq d-1$,
while $f_C(x)\geq 0$ and $g(x) \leq 0$ when $|x|\geq d$. 
Finally, as noted above $f_C(0)=1$.
This relation applies to all codes, and hence
\begin{equation}
A(n,d) \leq  g(0)/\widehat{g}(0) \text{~for every feasible~}g:\cube{n}\to\R
\label{eq:delsarte_dual}
\end{equation}

Note also that to turn the minimization 
problem in \eqref{eq:delsarte_dual}
into an honest LP we can just require that $\widehat{g}(0) = 1$.

\hfill

We describe next a general recipe that yields a feasible solution to \eqref{eq:delsarte_dual}. 
Take $g\coloneqq \phi \cdot \Gamma^2$, where
\begin{enumerate}
\item \label{condition_phi}$\phi:\cube{n}\to\R$ satisfies $\phi(x)\le 0$
    whenever $|x|\geq d$ and $\phi(0) > 0$, and
\item \label{condition_f} The nonzero function $\Gamma$ satisfies
    \[
        \widehat{\Gamma} \geq 0,~\widehat{\phi} * \widehat{\Gamma} \geq \widehat{\Gamma}
    \]
\end{enumerate}
Indeed, such $g$ is non-positive on levels $d,d+1,\dots,n$ of the cube, 
and, by the convolution theorem,
$\widehat{g}\geq \widehat{\Gamma} * \widehat{\Gamma} \geq 0$. 
In particular, $\widehat{g}(0) \geq \lVert \widehat{\Gamma} \rVert_2^2 > 0$.
We derive an upper bound on $A(n,d)$ in terms of the support size of $\widehat{\Gamma}$.
\begin{equation}
    \frac{g(0)}{\widehat{g}(0)} 
    \leq \phi(0) \frac{\Gamma^2(0)}{(\widehat{\Gamma}*\widehat{\Gamma}) (0)}
    = \phi(0) \frac{\langle \widehat{\Gamma}, \1_{\supp(\widehat{\Gamma})} \rangle^2}
    {\langle\widehat{\Gamma},\widehat{\Gamma}\rangle}
    \leq \phi(0)\cdot|\supp(\widehat{\Gamma})|
    \label{eq:value_bound}
\end{equation}
The last inequality follows from Cauchy-Schwartz.

Note that this method involves only $\widehat{\Gamma}$ with no direct reference to $\Gamma$.

This approach was spelled out in \cite{navon2005delsarte}. 
In fact, all constructions known to us, 
\cite{mceliece1977new, friedman2005generalized, navon2005delsarte, navon2009linear, samorodnitsky2021one, barg2006spectral, barg2008functional} 
follow this general pattern with $\phi(x) = 2(d-|x|)$.

\section{The Proof}

In this section we present our proof, which is cast in the pattern
described in the previous section.

Let $m,r$ be positive integers, and define
\begin{align}
    \phi(x) = \phi_m(x) &\coloneqq (n-2|x|)^m - (n-2d)^m
    \label{eq:our_sol_phi}
    \\
    \widehat{\Gamma} = \widehat{\Gamma}_r &\coloneqq L_{r} + L_{r-1}
    \label{eq:our_sol_f}
\end{align}
where $L_j$ is the indicator of the $j$-th level in
the Hamming cube,
\begin{equation}
    L_j(x)\coloneqq \1_{[|x|=j]}
    \label{eq:level_set_indicator}
\end{equation}

We establish the upper bound on $\mathcal{R}(\delta)$
by showing that $g=\phi_m \cdot \Gamma_r^2$ 
is feasible for suitable parameters $r,m$, 
and then optimizing over $r$
to minimize the value $g(0)/\widehat{g}(0)$.

\begin{proposition}\label{prop:feasible_solution}
    The function $g \coloneqq \phi_m \cdot \Gamma_r^2$ is a feasible
    solution to \eqref{eq:delsarte_dual}, provided that 
    $|r-n/2| \leq \sqrt{d(n-d)} - o(n)$,
    and $\omega(1)\leq m\leq o(n)$ is odd. 
\end{proposition}
For the rest of the section, we write $\phi,\Gamma$ instead of $\phi_m,\Gamma_r$.
\begin{proof}
Clearly, $\phi$ 
satisfies the first condition above, namely $\phi(x)\leq 0$ if $|x|\geq d$
and $\phi(0) > 0$.
It is also clear that $\widehat{\Gamma} \geq 0$. 
To show that $g = \phi \cdot \Gamma^2$ is feasible,
it remains to prove that
$\widehat{\phi} * \widehat{\Gamma} \geq \widehat{\Gamma}$.

As in \eqref{eq:level_set_indicator}, the indicator of the first level-set 
is denoted by $L_1(x)$. We note that the Fourier transform of $L_1(x)$
is $n-2|x|$. Indeed,
\[
    2^n \langle L_1, \chi_x\rangle 
    = \sum_{y\in\cube{n} : |y|=1} \chi_x(y)
    = \sum_{i=1}^{n}(-1)^{x_i}
    = (n-|x|)\cdot 1 + |x|\cdot (-1)
\] 
Let $A$ be the adjacency matrix of the Hamming cube,
$
    A(x,y) = \1_{[|x+y|=1]}
$.
Convolution with the function $L_1$ is equivalent to multiplication by $A$,
\[
    (L_1 * \widehat{\Gamma})(x) 
    = \sum_{y\in\cube{n} : |y|=1} \widehat{\Gamma}(x+y)
    = \sum_{z: |z+x|=1} \widehat{\Gamma}(z)
    = (A\widehat{\Gamma})(x)
\]
Namely,
$\widehat{\phi}*\widehat{\Gamma} = \left(A^m - (n-2d)^m I\right)\widehat{\Gamma}$.
We will achieve our goal by proving that
\begin{equation}
    A^m (L_r + L_{r-1}) \geq \left[(n-2d)^m + 1\right] (L_r + L_{r-1})
    \label{eq:main_goal}
\end{equation}
But $A^m(x,y)$ is the number of length-$m$ walks from $x$
to $y$ in the Hamming cube, and
\[
    A^m (L_r+L_{r-1})(x) = \sum_{y: |y|\in\{r,r-1\}} A^m(x,y)
\]
This is the number of walks that start at $x$ and terminate
in levels $\{r,r-1\}$.
Since $m$ is odd and the Hamming cube is bipartite, $A^m(x,y)$ is non-zero
only when $|x|$ and $|y|$ have different parity.
If $|x|\notin \{r,r-1\}$, the right-hand side of \eqref{eq:main_goal} 
is $0$ and the left-hand side is non-negative. 

Denote by $P_{r,m,j}$ the number of length-$m$ walks on the Hamming cube
that start at a given vertex in level $r$ and end in level $r+j$. Then,
\begin{equation}
    A^m (L_r+L_{r-1})(x)
    \geq P_{r,m,-1}L_{r}(x) + P_{r-1,m,1} L_{r-1}(x) 
    \label{eq:bound_by_num_walks}
\end{equation}
The following proposition gives an asymptotic estimate of
$P_{r,m,j}$.
\begin{proposition}\label{prop:num_walks}
    Let $r, m$, and $j$ be integers such that
    \begin{itemize}
        \item $r$ is linear in $n$, $r = \Theta(n)$,
        \item $m$ is a slowly growing function of $n$, $\omega(1) \leq m \leq o(n)$,
        \item $|j| = o(m)$ and $m\equiv j\bmod 2$. 
    \end{itemize}
    Then,
    \[
        (P_{r,m,j})^{1/m} = 2\sqrt{r(n-r)} \pm o(n)
    \]
\end{proposition}
\begin{proof}[Proof of Proposition \ref{prop:num_walks}]
    A vertex at level $k$ has $n-k$ neighbors at level $k+1$, and
    $k$ neighbors at level $k-1$. 
    Therefore, whenever the walk takes a step up
    (i.e., goes to the next higher level) it can do so in
    $n-r \pm O(m)$ ways, since the whole walk is 
    contained in levels $r-m,\ldots,r+m$.
    Likewise, there are $r \pm O(m)$ options at each down step.
    
    A path that starts at level $r$ and ends at level $r+j$ 
    takes $(m+j)/2$ steps up and $(m-j)/2$
    steps down.  
    There are $\binom{m}{(m-j)/2}=2^m/m^{O(1)}$ ways to choose the order of the
    up/down steps. Consider any word in the letters "up" and "down" with $(m+j)/2$
    "up"s and $(m-j)/2$ "down"s. The number of corresponding walks in the cube
    is the product of $(m+j)/2$ terms each of which equals \mbox{$n-r \pm O(m)$},
    and $(m-j)/2$ terms each of which equals $r \pm O(m)$. Therefore,
\[
        P_{r,m,j}  = \frac{2^m}{m^{O(1)}}
            \left[n-r \pm O(m))\right]^{(m+j)/2}
            \left[r\pm O(m))\right]^{(m-j)/2}=
\]
\[
= \left(2\sqrt{r(n-r)}\right)^m 
\left(\frac{n-r \pm O(m)}{n-r}\cdot\frac{r \pm O(m)}{r}\right)^{m/2}
        \left(\frac{n-r \pm O(m)}{r \pm O(m)}\right)^{j/2}
        m^{-O(1)}
\]
    Taking the $m$-th root, we obtain the main term $2\sqrt{r(n-r)}$,
    while the other terms contribute a factor of $1\pm o_n(1)$.
\end{proof}

By \eqref{eq:bound_by_num_walks} and the above Proposition,
\begin{equation*}
    A^m (L_r+L_{r-1})(x)
    \geq \left(2\sqrt{r(n-r)} - o(n)\right)^{m} (L_{r} + L_{r-1})(x) 
\end{equation*}
So \eqref{eq:main_goal} is satisfied if
\[
    2\sqrt{r(n-r)} - o(n) \geq ((n-2d)^m+1)^{1/m}
\]
namely if
\begin{equation}
    |r-n/2| \leq \sqrt{d(n-d)} - o(n)
    \label{eq:condition_on_r}
\end{equation}
Choosing any $r$ in this range guarantees that $\widehat{\phi}*\widehat{\Gamma}\geq \widehat{\Gamma}$, 
and consequently that $g = \phi \cdot \Gamma^2$
is feasible.
\end{proof}

By \eqref{eq:value_bound}, this shows that $A(n,d) \leq \phi_m(0) |\supp (\widehat{\Gamma}_r)| \le 2n^m \binom{n}{r}$,
for $r$ that satisfies \eqref{eq:condition_on_r}.
To make the first factor sub-exponential, we take $m=o(n/\log n)$,
and consequently $A(n,d) \le 2^{nH(r/n) + o(n)}$. The term $\binom{n}{r}$
is smaller the further from $n/2$ that $r$ is. So, we choose $r = n/2 - \sqrt{d(n-d)}+o(n)$
in order to minimize the bound and conclude that
\[
    A(n,\lfloor \delta n \rfloor)
    \leq 2^{n H(1/2 - \sqrt{\delta(1-\delta)}) + o(n)}
\]
which completes the proof of \eqref{eq:first_lp_bound}.

\section{Related Work}

There are several proofs of first LP bound
in the literature 
\cite{mceliece1977new, friedman2005generalized, navon2005delsarte, navon2009linear, samorodnitsky2021one, barg2006spectral, barg2008functional}. These various proofs
originate from different perspectives, and each of them
offers new interpretations and insights. 
Essentially, they all follow the same pattern
that we introduced in Section \ref{section:delsarte_lp},
with
\[
    \phi_{MRRW}(x) = 2(d-|x|)
\]
and $\Gamma$ in the spirit of
\[
    \Gamma_{MRRW}(x) = \sum_{j=0}^{r} \binom{n}{j}^{-1} K_{j}(d)K_j(|x|)
\]
where $\{K_0,\dots,K_n\}$ are the Krawtchouk polynomials. Here
\begin{equation}
    K_{j}(t) = \sum_{i=0}^{n}(-1)^{i} \binom{t}{j}\binom{n-t}{k-j}
    \quad
    t \in \R
    \label{eq:krawtchouk}
\end{equation}
is the Fourier transform of the indicator function of the $j$-th level of the cube.

The proof of \cite{mceliece1977new} uses the Christoffel-Darboux
formula from the theory of orthogonal polynomials. The parameter $r$
is determined so that the first root of the 
$r$-th Krawtchouk is at least $n-2d$, which leads, in turn to 
\eqref{eq:condition_on_r}. In \cite{friedman2005generalized,navon2005delsarte,navon2009linear,samorodnitsky2001optimum},
$\widehat{\Gamma}_{MRRW}$ 
is chosen as the eigenfunction
of the Hamming ball of radius $r$ that corresponds to the largest 
eigenvalue (the Perron eigenfunction). 
In \cite{barg2006spectral,barg2008functional}, a function close to
$\Gamma_{MRRW}$ is derived using the Perron-Frobenius theorem and a 
functional-analytic perspective.

\section*{Acknowledgements}
We are thankful to Leonardo Nagami Coregliano, Fernando Granha Jeronimo and Chris Jones for insightful discussions.
We also thank Alex Samorodnitsky for 
helpful advice and inspiration.

\bibliographystyle{plain}
\bibliography{refs}

\end{document}